\documentclass{article}
\def\={\>{\buildrel \rm def \over =}\>}

\def\S{{\cal S}}        
\def\HH{\mathop{\hbox{\bf H}}} 
\def\qed{\vrule width 0.5em height 2ex depth 0pt}
\newtheorem{theorem}{Theorem}[section]
\newtheorem{lemma}[theorem]{Lemma}
\newtheorem{fact}[theorem]{Fact}

\newtheorem{corollary}[theorem]{Corollary}
\newtheorem{DDDefinition}[theorem]{Definition}
\def\enddefinition{\end{DDDefinition}\egroup\medbreak}
\def\proof{\begin{demo}{Proof}}
\def\endproof{\enspace\hfill\qed\end{demo}\medbreak}
\newenvironment{demo}[1]
{\par\medbreak\noindent{\bf #1\enskip}\rm\ignorespaces}{\par}
\makeatletter
\def\definition{\bgroup\def\@begintheorem##1##2{\trivlist
  \item[\hskip\labelsep{\bfseries ##1\ ##2}]}\begin{DDDefinition}}
\def\@listI{\leftmargin\leftmargini \parsep 1pt plus 1pt minus 1pt\topsep 1pt
plus 1pt minus 1pt\itemsep 1pt plus 1pt minus 1pt}
\let\@listi\@listI
\@listi
\def\@listii{\leftmargin\leftmarginii
 \labelwidth\leftmarginii\advance\labelwidth-\labelsep
 \topsep 1pt plus 1pt minus 1pt
 \parsep 1pt plus 1pt minus 1pt
 \itemsep \parsep}
\makeatother
\def\eqalign#1{\null\,\vcenter{\openup\jot\mathsurround=0pt
  \ialign{\strut\hfil$\displaystyle{##}$&$\displaystyle{{}##}$\hfil
      \crcr#1\crcr}}\,}
\def\[{\mathopen{\hbox{$[\![$}}}
\def\]{\mathclose{\hbox{$]\!]$}}}
\def\minus#1{\mathbin{\mbox{--}}\{#1\}}

\author{L\'aszl\'o Csirmaz\\
\small Central European University\thanks{
This research was partially supported by TAMOP-4.2.2.C-11/1/KONV-2012-0001 
and the Lendulet Program of the Hungarian Academy of Sciences}}
\date{}
\title{\bf Secret sharing on the $d$-dimensional cube}

\begin{document}
\maketitle
\begin{abstract}
\noindent
We prove that for $d>1$ the best
information ratio of the perfect secret sharing scheme
based on the edge set of the $d$-dimensional cube is exactly $d/2$. 
Using the technique developed, we also prove that the information ratio
of the infinite $d$-dimensional lattice is $d$.


\end{abstract}

\section{Introduction}
In a (perfect) secret sharing scheme, a secret value is distributed in the
form of {\em shares} among the set of participants in such a way that only
qualified sets of participants can recover the secret value, while no
information about the secret is revealed by the collective share of an
unqualified subset. Consult the survey of A.~Beimel \cite{beimel-survey} 
for a general overview, or the lecture notes of C.~Padro 
\cite{padro:lecture-notes} for a gentle introduction to the topic.

The {\em information ratio} of a scheme is the ratio between the maximum
size of the shares, and the size of the secret value, while the information
ratio of the collection of qualified subsets -- the {\em access structure}
-- is the infimum of the information ratio of schemes realizing this access
structure. (In the literature the term 
{\em complexity} is also used to denote the information ratio.)
One of the main theoretical and practical problems of this area is
to determine, or give reasonable bounds for, the information ratio of
different access structures. 

An access structure is {\em graph based} when the minimal qualified subsets are just
the edges (two element subsets) of a (connected) graph with participants 
as vertices. The access structure determined by the complete graph on $n$ vertices 
is the 2-out-of-$n$ threshold structure: any two element subset is
qualified, and any single element subset is unqualified.

Determining the exact value of the information ratio 
of arbitrary graphs is a very difficult problem.
It has been determined for most of the graphs with at most
six vertices \cite{brickell-stinson, kn:blundo, dijk, jackson}, and for the 
majority of graphs with seven vertices
\cite{seven}. The exact ratio is also known for a couple of infinite 
graph families.
For example, complete graphs have information ratio 1;
paths on four or more vertices as well as cycles of length at least 5 have
information ratio 3/2 \cite{padro:lecture-notes}. Trees have
information ratio $2-1/k$ for some easily computable integer $k$
\cite{csirmaz-tardos}; and graphs with girth ${}\ge 6$ and no
neighboring $\ge3$ degree vertices also have ratio $2-1/k$
for some integer $k$ 
\cite{csirmaz-ligeti}.

Brickell and Davenport in \cite{brickell-davenport} proved that a graph has
information ratio 1 if and only if it is a complete multipartite graph.
The information ratio is bounded from above by
$(d+1)/2$, where $d$ is the maximal degree, see \cite{kn:stinson}. In
\cite{blundo-tight} Blundo et al. constructed, for each $d\ge 2$, and
infinite family of $d$-regular graphs with exactly this ratio. As the
maximal degree of a graph on $n$ vertices is at most $n-1$, the information
ratio
of any graph on $n$ vertices is at most $n/2$. This upper bound was improved 
to $c\cdot n/\log_2 n$ for some large explicit constant $c$ in \cite{erdos-pyber}, 
which is still the best
general upper bound on the information ratio of graphs on $n$ vertices. 

The smallest $d$-regular graph with information ratio $(d+1)/2$ in \cite{blundo-tight}
has $n\approx 6^d$ vertices which shows that for some graphs the information
ratio is at least $c'\cdot\log_2 n$ for some small positive constant $c'$.
In \cite{kn:csgraph} a quite natural example for
a $d$-regular graph was considered: the edge graph of the $d$-dimensional
cube, giving upper and lower bounds for the information ratio. Determining the
exact value, however, remained an open problem. In this note, using a
carefully crafted induction hypothesis, we show that this information ratio
is exactly $d/2$.
\begin{theorem}\label{thm:0}
The information ratio of the edge graph of the $d\ge 2$ dimensional cube is
$d/2$.
\end{theorem}
As the $d$-dimensional cube has exactly $n=2^d$ vertices, this theorem
yields the improved lower bound in the corollary below:
\begin{corollary}
There is an explicit positive constant $c>0$ such that
for infinitely many $n$ the largest possible information ratio of a graph on
$n$ vertices is between $0.5\cdot\log_2n$ and $c\cdot n/\log_2 n$.
\end{corollary}

There is huge gap between the lower and upper bounds. It is an open problem
to narrow this gap. It is interesting to note that the lower bound comes
from a sparse graph (the maximal degree is $o(n)$), while the upper bound
requires dense ($\Omega(n)$ average degree), but not very dense
($n-o(n)$ average degree) graphs, see \cite{beimel-farras-mintz}.

Using the same technique as in the proof of Theorem \label{thm:8}
we can also determine the information ratio 
of the whole $d$-dimensional lattice $L^d$, which was also left open in
\cite{kn:csgraph}.
\begin{theorem}\label{thm:01}
For $d\ge2$ the information ratio of the $d$-dimensional lattice $L^d$ is
$d$.
\end{theorem}

\bigskip

This paper is organized as follows. In Section \ref{s:defs} we give the
definitions necessary to state and prove our theorems. Section \ref{s:cube}
deals with the case of the $d$-dimensional cube, Section \ref{s:lattice}
with the lattice. Finally Section \ref{s:conclusion} concludes the paper,
and lists some related problems. For undefined notions and for
introduction to secret sharing consult \cite{beimel-survey, kn:blundo,
padro:lecture-notes}, and
for basics in information theory see \cite{kn:csiszar}.

\section{Definitions}\label{s:defs}

In this section we recall the notions we shall use later. First we 
give a formal definition of a graph based
a perfect secret sharing scheme, then connect it to
submodular functions.

Let $G=\langle V,E\rangle$ be a graph with vertex set $V$ and edge set $E$.
A subset $A$ of $V$ is {\em independent} if there is no edge between
vertices in $A$. A {\em covering} of the graph $G$ is a collection of
subgraphs of $G$ such that every edge is contained in one of the (not
necessarily spanned) subgraphs in the collection. The collection is {\em
$k$-covering} if every edge of $G$ is covered at least $k$ times.
For subsets of vertices we usually omit the $\cup$ sign, and write $AB$ for
$A\cup B$. Also, if $v\in V$ is a vertex then $Av$ denotes $A\cup\{v\}$. 

A {\em perfect secret sharing scheme $\S$} for a graph $G$ is a collection
of random variables $\xi_v$ for each $v\in V$ and the random variable $\xi$ (the secret)
with a joint distribution of $\langle \xi_v\rangle$ and $\xi$ so that 
\begin{itemize}
\item if $vw$ is an edge in $G$, then $\xi_v$ and $\xi_w$ together determine
     $\xi$,
\item if $A$ is an independent set, then $\xi$ and the collection
$\{\xi_v\,:\,v\in A\}$ are statistically independent.
\end{itemize}
The {\em size} of the random variable $\xi$ is measured by its entropy, or
information content, and is denoted by $\HH(\xi)$, see \cite{kn:csiszar}.
The {\em information ratio} for a vertex $v$ (participant)
of the graph $G$ is $\HH(\xi_v)/\HH(\xi)$.
This value tells how many bits of information $v$ must remember for 
each bit in the secret. The {\em worst case} and {\em average} information
ratio of $\S$ are the highest and average information ratio among all 
participants, respectively.

The worst case (average) information ratio of a graph $G$ is the
infimum of the worst case (average) information ratio of all perfect secret 
sharing schemes
$\S$ defined on $G$. 

\smallskip
Let $\S$ be a perfect secret sharing scheme based on the graph $G$ with
the random variable $\xi$ as secret, and $\xi_v$ for $v\in V$ as shares.
For each subset $A$ of the vertices let us define 
$$
       f(A) \= { \HH( \{ \xi_v \,:\, v\in A\} ) \over \HH(\xi) }.
$$
Clearly, the average information ratio of $\S$ is the average of $\{ f(v)\,:\, 
v\in V\}$, and the worst case information ratio is the maximal value in this 
set. Using standard properties of the entropy function,
cf.~\cite{kn:csiszar}, we have
\begin{itemize}
\item[(a)] $f(\emptyset)=0$, and in general $f(A)\ge 0$ (positivity);
\item[(b)] if $A\subseteq B \subseteq V$ then $f(A)\le f(B)$ (monotonicity);
\item[(c)] $f(A) + f(B) \ge f(A\cap B) + f(A\cup B)$ (submodularity).
\end{itemize}
For two random variables $\eta$ and $\xi$, the value
of $\eta$ determines the value of $\xi$ iff $\HH(\eta\xi) = \HH(\eta)$, and 
$\eta$ and $\xi$ are (statistically) independent iff $\HH(\eta\xi) = 
\HH(\eta) + \HH(\xi)$. Using these facts and the definition of the 
perfect secret sharing scheme, we also have
\begin{itemize}
\item[(d)] if $A\subseteq B$, $A$ is an independent set and $B$ is not, then
   $f(A) + 1 \le f(B)$ (strong monotonicity);
\item[(e)] if neither $A$ nor $B$ is independent but $A\cap B$ is so, then
  $f(A)+f(B) \ge 1+f(A\cap B)+f(A\cup B)$ (strong submodularity).
\end{itemize}
The so-called entropy method can be rephrased as 
follows. Prove that for {\em any} real-valued function $f$ satisfying 
properties (a)--(e) the average (or largest) value of $f$ on the vertices 
is at least $\rho$. Then, as functions coming from secret sharing
schemes also satisfy these properties, conclude, that the average (or worst
case) information ratio of $G$ is also at least $\rho$. We note that this
method is not universal, as properties (a)--(c) are too weak to
capture exactly the functions coming from entropy.

\smallskip

We frequently use the submodular (c) and the strong submodular (e) properties 
in the following rearranged form whenever $A$, $X$, and $Y$ are disjoint
subsets of the vertex set $V$:
\begin{itemize}
\item[(c$'$)] $f(AX) - f(A) \ge f(AXY) - f(AY)$;
\end{itemize}
moreover, if $A$ is independent (i.e.~empty), $AX$ and $AY$ are not, then
\begin{itemize}
\item[(e$'$)] 
$f(AX) - f(A) \ge f(AXY) - f(AY) + 1$.
\end{itemize}
In particular, if both $X$ and $Y$ contain an edge (and they are disjoint), then 
$f(X) \ge f(XY) - f(Y) + 1$.

The proof of the following easy fact is omitted:\medbreak

\begin{fact}\label{fact:spanned}
Suppose $G_2$ is a spanned subgraph of $G_1$. The worst case (average)
information ratio of $G_1$ is at least as large as the worst case (average)
information ratio of $G_2$. 
\hfill\qed
\end{fact}

\section{The case of the cube}\label{s:cube}

The {\em $d$-dimensional cube}, denoted here by $C^d$, is the following graph. 
Its vertices
are 0--1 sequences of length $d$. Two vertices are connected by an edge if
the sequences differ in exactly one place. This cube can be embedded into
the $d$-dimensional Euclidean space. Points with all coordinates in the set
$\{0,1\}$ are the vertices, and two vertices are connected if their
distance is $1$. 

The $d$-dimensional cube has $2^d$ vertices, $d\cdot 2^{d-1}$ edges
(they correspond to one-di\-men\-sion\-al affine subspaces in the embedding),
and each vertex has degree $d$.
The two-di\-men\-sion\-al subspaces are squares, i.e.~cycles of
length four, we call them {\it $2$-faces}. Each vertex $v$ 
is adjacent to $d\choose 2$ such $2$-face, as any pair of edges starting 
from $v$ spans a $2$-face. Consequently the number of $2$-faces is $2^{d-2} 
{d \choose 2}$. For any edge there are exactly $(d-1)$ many $2$-faces 
adjacent to that edge. It means that $2$-faces, as subgraphs, constitute
a $(d-1)$-cover of $C^d$.

\medbreak
\begin{theorem}\label{thm:1}
The information ratio of the $d \ge 2$ dimensional cube is $d/2$.
\end{theorem}

\medbreak\noindent
We note that this statement is not true for $d=1$. The 1-dimensional 
``cube'' is the graph
with two vertices and an edge between them. In this graph both the worst 
case and average information ratio is equal to $1$, and not to $1/2$. The
$2$-di\-men\-sion\-al ``cube'' is the square, i.e.~a cycle on four vertices, 
which is a complete bipartite graph. Thus both worst case and average
information ratio of $C^2$ is $1$, in full agreement with the statement.

\begin{proof}
First we prove that this ratio is at most $d/2$. To this end we construct a
perfect secret 
sharing scheme witnessing this value. The construction uses Stinson's
decomposition theorem from \cite{kn:stinson}.

\def\s{\mathbf{s}}
\def\x{\mathbf{x}}

Let $F$ be a sufficiently large finite field, and $X$ be the
$(d-1)$-di\-men\-sion\-al vector space over $F$. For every $2$-face of the cube
choose a vector $\x_i\in X$ in such a way that any $d-1$ of these vectors span the 
whole vector space $X$. (This
is the point where we use the fact that $F$ is sufficiently large.) The
vectors $\x_i$ are public information, and the
secret is a random element $\s\in X$. For each vector $\x_i$ take the
inner product $a_i = \s\cdot\x_i$. Clearly, given any $(d-1)$ of these inner
products, one can
recover the secret $\s$. Now suppose the $i$-th $2$-face has vertices $v_1$,
$v_2$, $v_3$, $v_4$ in this order. Distribute $a_i$ among these vertices
as follows.
Choose a random element $r\in F$ and give it to $v_1$ and $v_3$, and
give $r+a_i$ (computed in the field $F$) to $v_2$ and $v_4$. Any edge of
this
$2$-face can recover $a_i$, thus any edge of the $d$-dimensional cube can
recover $d-1$ of the $a_i$'s, and therefore can recover the secret $\s$ as well.
Now consider the values an independent set of the vertices possess. All
different values in this set can be chosen independently and randomly from
$F$, thus they are (statistically) independent of the secret $\s$. 

We have verified that this is a perfect secret sharing system. The 
secret is
a $(d-1)$-tuple from the field $F$. Each vertex is given as many elements
from $F$ as many $2$-faces it is in, namely $d\choose 2$ elements. Therefore 
both worst case and average information ratio for this scheme is ${d\choose 2}/(d-1) = d/2$, which
proves the upper bound.

\smallskip

Before handling the lower bound, observe that the worst case and the average case information ratio 
for cubes must coincide. The reason is that $C^d$
is highly symmetrical. Let $H$ be the automorphism group of the graph $C^d$, 
this group has $2^d\cdot d!$ elements. If $v_1$ and $v_2$ are two (not
necessarily different) vertices of $C^d$, then the number of 
automorphisms $\pi\in H$ 
with $\pi(v_1)=v_2$ is exactly $|H|/|C^d| = d!$. Now let $\S$ be 
any perfect secret sharing scheme on $C^d$, and
apply $\S$ for $\pi C^d$ independently for each $\pi\in H$. The
size of the secret in this compound scheme increases $|H|$-fold, 
and each participant will get 
a share which has size $|H|/|C^d|$-times the sum of all share sizes in $\S$. 
Therefore in this ``symmetrized'' scheme all participants have the same 
amount of information to remember, consequently all have the same ratio 
which equals to the average ratio of the scheme $\S$.

\smallskip
Thus to prove that $d/2$ is also a lower bound for both the worst case and
average information ratio of $C^d$ it is enough to show that for any real valued
function $f$ satisfying properties (a)--(e) enlisted in section \ref{s:defs}
we have
$$
      \sum \{ f(v) \,:\, v\in V \} \ge {d\over 2}|V| .
$$
This is exactly what we will do.

Split the vertex set of the $d$-dimensional cube $C^d$ into two 
equal 
parts in a ``chessboard-like'' fashion: $C^d=A_d\cup B_d$, 
where $A_d$ and $B_d$ are disjoint, independent, and
$|A_d| = |B_d| = 2^{d-1}$. Vertices in $A_d$ only have neighbors in
$B_d$, and vertices in $B_d$ only have neighbors in $A_d$. The
$(d+1)$-dimensional cube consist of two disjoint copy of the $d$ dimensional
cube at two levels, and there is a perfect matching between the
corresponding vertices. Each edge of $C^{d+1}$ 
is either a vertex of one of the lower dimensional cubes, or is a member 
of the perfect matching. Suppose the vertices 
on these two smaller cubes are split as $A_d\cup B_d$ and $A'_d \cup B'_d$,
respectively, such that the perfect matching is between $A_d$ and $B'_d$, 
and between $B_d$ and $A'_d$. Then the splitting of the vertices of 
the $(d+1)$-dimensional cube can be done as
$$
    A_{d+1} = A_d \cup A'_d \mbox{~~~and~~~} B_{d+1}=B_d\cup B'_d .
$$
Using this decomposition, we can use induction on the dimension $d$. In the 
inductive statement we shall use the following notation:
$$
    \[ A,B \] \= \sum_{b\in B} f(bA) - \sum_{a\in A} f(A\minus a).
$$
When using this notation we implicitly assume that $A$ and $B$ have the same cardinality.

\begin{lemma}\label{lemma:1}
For the $d$-dimensional cube with the split $C^d=A_d\cup B_d$ we have
\begin{equation}\label{eq1}
    \sum_{v\in C^d} f(v) \ge \[ A_d, B_d \]  + (d-1)2^{d-1}.
\end{equation}
\end{lemma}
\begin{proof}
First check this inequality for $d=1$. The $1$-cube has two
connected vertices $a$ and $b$. Then, say, $A_1=\{a\}$, $B_1=\{b\}$, 
and equation (\ref{eq1}) becomes
$$
     f(a)+f(b) \ge  f(ab) - f(\emptyset) + 0,
$$
which holds by the submodular property (c) of the function $f$.

Now suppose (\ref{eq1}) holds for both $d$-dimensional subcubes of the
$(d+1)$-di\-men\-sion\-al cube with split $A_{d+1}=A_d\cup A'_d$, and 
$B_{d+1}=B_d\cup B'_d$ as discussed above. Then by the inductive hypothesis,
\begin{eqnarray}
   \sum_{v\in V_{d+1}} f(v)
   &=& \sum_{v\in V_d} f(v) + \sum_{v'\in V'_d} f(v') \nonumber\\
   &\ge& \[ A_d,B_d\] + \[ A'_d, B'_d \] + (d-1)2^d . \label{eq:todo} 
\end{eqnarray}
Each $b\in B_d$ is connected to a unique $a'\in A'_d$, let $(a',b)$ be
such a pair. Then
\begin{equation}\label{eq2}
       f(bA_d) - f(A_d) \ge f(bA_dA'_d\minus{a'}) - f(A_dA'_d\minus{a'})
\end{equation}
by submodularity. Now let $a\in A_d$ be any vertex which is
connected to $b\in B_d$. As $b$ is connected to both $a$ and $a'$, 
both $bA'_d$ and $abA'_d\minus{a'}$ are 
qualified (i.e.~not independent) subsets, while their intersection,
$bA'_d\minus{a'}$, is independent. Therefore the strong submodularity yields
$$
     f(bA'_d) - f(bA'_d\minus{a'}) \ge 1+ f(baA'_d) - f(baA'_d\minus{a'}).
$$
Using this inequality and the submodularity twice we get
\begin{eqnarray*}
   f(A'_d) - f(A'_d\minus{a'}) &\ge& f(bA'_d) - f(bA'_d\minus{a'}) \\
   &\ge& 1+   f(baA'_d)-f(baA'_d\minus{a'}) \\
   &\ge& 1+ f(bA_dA'_d)-f(bA_dA'_d\minus{a'}).
\end{eqnarray*}
Adding (\ref{eq2}) to this inequality, for each connected 
pair $(a',b)$ from $a'\in A'_d$ and $b\in B_d$ we have
$$
  f(bA_d)-f(A_d)+f(A'_d)-f(A'_d\minus{a'}) \ge
       1 +  f(bA_dA'_d) - f(A_dA'_d\minus{a'}) .
$$
By analogy we can swap $(A_d,B_d)$ and $(A'_d,B'_d)$ yielding 
$$
  f(b'A'_d) - f(A'_d) + f(A_d) - f(A_d\minus a) \ge
       1 + f(b'A_dA'_d) - f(A_dA'_d\minus a) 
$$
for each connected pair $(a,b')$ from $a\in A_d$ and $b'\in B'_d$.
There are $2^{d-1}$ edges between $A'_d$ and $B_d$, and also $2^{d-1}$
edges between $A_d$ and $B'_d$. Thus adding up all of these $2^d$
inequalities, on the left hand side all $f(A_d)$ and $f(A'_d)$
cancel out, and the remaining terms give
$$
    \[ A_d,B_d \] + \[ A'_d, B'_d \] \ge \[ A_dA'_d, B_dB'_d \] + 2^d .
$$
Combining this with (\ref{eq:todo}) we get
$$
 \sum_{v\in V_{d+1}} f(v) \ge
   \[ A_dA'_d , B_dB'_d \] + (d-1)2^d + 2^d .
$$
This is inequality (\ref{eq1}) written for $d+1$ instead of $d$. This
completes the induction step.
\end{proof}

We continue with the proof of theorem \ref{thm:1}. Let $C^d = A_d \cup B_d$ be 
the disjoint ``chessboard'' splitting of the vertices. 
As there are exactly $2^{d-1}$ vertices in both $A_d$ and
$B_d$, we can match them. If $(a,b)$ is such a matched pair, then by
strong monotonicity
$$
      f(bA_d) - f(A_d\minus a) \ge 1,
$$ 
as $A_d\minus a$ is independent, while $bA_d$ is not. Adding up these 
inequalities we get
$$
     \[ A_d, B_d \] =  \sum_{b\in B_d} f(bA_d) - \sum_{a\in A_d} f(A_d\minus a) \ge 2^{d-1}.
$$
This, together with the claim of Lemma \ref{lemma:1} gives
$$
      \sum_{v\in V_d} f(v) \ge (d-1)2^{d-1} + 2^{d-1} = d2^{d-1}.
$$
There are $2^d$ vertices in $V_d$, thus the average value of $f$ on 
the vertices of $V_d$ is at least $d/2$. This shows that the average 
information ratio of the $d$-dimensional cube is at least $d/2$. From this
it follows that the worst case information ratio is also at least $d/2$.
\end{proof}

\section{The case of the lattice}\label{s:lattice}

The vertices of the {\em $d$-dimensional lattice $L^d$} are
the integer points of the $d$-di\-men\-sion\-al Euclidean space,
i.e.~points having integer coordinates only. Two vertices are connected 
if their distance is exactly $1$, i.e.~if they differ in a single 
coordinate, and the difference in that coordinate is exactly 1. Of course,
$L^d$ is an infinite graph.

Each vertex in $L^d$ has degree $2d$, and the whole graph is {\em edge transitive}.
Namely, given any two edges $v_1v_2$ and $w_1w_2$ from $L^d$, there is an 
automorphism of $L^d$
which maps $v_1$ to $w_1$ and $v_2$ to $w_2$.

Defining information ratio for an infinite graph is not straightforward.
A systematic treatment of the topic can be found in
\cite{csirmaz:tatramountains}. We remark that using the right
definitions all intuitively true statements remain true, among others 
Stinson's decomposition theorem \cite{kn:stinson}.

As $L^1$ is the infinite path, its ratio is $3/2$. For larger dimensions we
have

\begin{theorem}\label{thm:2}
For $d\ge 2$ the information ratio of the $d$ dimensional lattice $L^d$ 
is $d$.
\end{theorem}

\begin{proof}
First we show that $d$ is an upper bound. This requires a construction of
a perfect secret sharing scheme 
in
which every vertex should remember at most $d$ times as much information 
as there is in the secret. Let $v$ be a vertex of $L^d$ whose all 
coordinates have the same parity -- i.e.~either all are odd or 
all are even integers. Increase each coordinate of $v$  
either by $0$ or $1$. The resulting $2^d$ points form a $d$-di\-men\-sion\-al
cube. Consider all of these cubes. They fill the whole space in a
chessboard-like fashion. Each vertex of $L^d$ belongs to exactly two such
cubes: one starting form a point with even coordinates only, and one starting 
from a point with odd coordinates only. Furthermore 
each edge of $L^d$ belongs to exactly one of these cubes.

Distribute the secret in each of these (infinitely many) cubes 
independently. By Theorem \ref{thm:1} this can be done so that each vertex 
of the cube gets exactly $d/2$ bits for each bit in the secret. As each
vertex in $L^d$ is in exactly two cubes, each vertex gets two times
$d/2$ bits. And as each vertex of $L^d$ is a vertex in some cube, endpoints
of a vertex can recover the secret.

The distribution of the shares in each cube was made by a perfect system,
and random values were chosen independently for each cube. Therefore 
independent
subsets of $L^d$ have no information on the secret. This proves that
$d$ is an upper bound for both the average and worst case information ratio.

Proving that $d$ is also a lower bound first we prove a generalization
of Lemma \ref{lemma:1}. To
describe the setting, suppose we have a graph with vertices split into six
disjoint sets $(A\cup A^*) \cup (B\cup B^*)\cup (A'\cup B')$. Subsets
$A\cup A^*\cup A'$ and $B\cup B^*\cup B'$ are independent, 
the cardinality of the subsets $A$, $A'$, $B$, and $B'$ are equal, furthermore
$|A^*| = |B^*|$.
There are $|A|=|B'|$ many edges between $A$ and $B'$ which form a perfect
matching, and there are $|A'|=|B|$ many edges between $A'$ and $B$ which
also form a perfect matching. All other edges of the graph are connecting
two vertices either from $A\cup A^*$ and $B\cup B^*$, or from $A'$ and $B'$.
This means, for example, that each $a'\in A'$ is connected to
exactly one member of $B$, and there is no edge, for example, between 
$B'$ and $A^*$.

\begin{lemma}\label{lemma:2}
With the notation above, let $|A| = |B| = |A'| = |B'| = k$. Suppose
moreover that each $b\in B$ is connected to some $a\in A\cup A^*$, and 
each $b'\in B'$ is connected to some $a'\in A'$. Then
$$
\[ AA^*, BB^* \] + \[ A', B' \] \ge 2k + \[ A'AA^* , B'BB^* \].
$$
\end{lemma}

\begin{proof}
As in the proof of Lemma \ref{lemma:1}, for $b\in B$ let $a'\in A'$ 
be the
only vertex connected to in $A'$, and let $a\in A\cup A^*$ which
$b$ is connected to as well. Then using submodularity and strong
submodularity,
$$
    f(bAA^*) - f(AA^*) \ge f(bAA^*A'\minus {a'}) - f(AA^*A'\minus{a'}),
$$
and
\begin{eqnarray*}
    f(A')-f(A'\minus {a'}) &\ge& f(bA') - f(bA'\minus{a'}) \\
      &\ge& 1+ f(baA') - f(baA'\minus{a'}) \\
      &\ge& 1+ f(bAA^*A') - f(bAA^*A'\minus{a'}).
\end{eqnarray*}
On the other hand, if $b'\in B'$ is connected to $a\in A$, and $a'\in A'$,
then
$$
    f(b'A') - f(A') \ge f(b'A'A^*A\minus a) - f(A'A^*A\minus a),
$$
and
\begin{eqnarray*}
    f(AA^*) - f(AA^*\minus a) &\ge& f(b'AA^*)-(b'AA^*\minus a) \\
      &\ge& 1+ f(b'a'AA^*) - f(b'a'AA^*\minus a) \\
      &\ge& 1 + f(b'A'AA^*) - f(b'A'AA^*\minus a).
\end{eqnarray*}
Summing up all of these inequalities, $2k$ in total, $f(AA^*)$ and $f(A')$ are
canceled out, and we get
$$\eqalign{
& \Big( \sum_{b\in B} f(bAA^*) - \sum_{a\in A} f(AA^*\minus a) \Big)
   + \Big( \sum_{b'\in B'} f(b'A') - \sum_{a'\in A'} f(A'\minus{a'}) \Big)
\cr
&~~~~~ \ge 2k + \sum_{b\in B\cup B'} f(bAA^*A') - 
       \sum_{a\in A\cup A'} f(AA^*A'\minus a) .
}$$
The missing part, namely that
$$
   \sum_{b\in B^*} f(bAA^*) - \sum_{a\in A^*} f(AA^*\minus a)
      \ge \sum_{b\in B^*} f(bAA^*A') - \sum_{a\in A^*} f(AA^*A'\minus a)
$$
follows immediately from submodularity and from $|A^*|=|B^*|$.
\end{proof}

As we will use Lemma \ref{lemma:2} inductively, we need to consider
the base case first, namely when the dimension is $1$. The 1-di\-men\-sion\-al
lattice is an infinite path; we handle its finite counterparts. Thus
let $k\ge 2$ be an even number, and let $a_1$, $b_1$, $\dots$, $a_{k/2}$,
$b_{k/2}$ be the vertices, in this order,  of a path of length $k$. Let
$A$ be the set of odd vertices, and $B$ be the set of even vertices.

\begin{lemma}\label{lemma:base}
For each path $P$ of even length $k\ge 2$, 
\begin{equation}\label{eq:lemmabase}
\sum_{v\in P} f(v) \ge \[ A, B\] + {k\over 2}-1.
\end{equation}
\end{lemma}

\begin{proof}
By induction on the length of the path. When $k=2$, i.e.~the
graph consists of two connected vertices $a$ and $b$ only, then by
submodularity
$$
   f(a)+f(b) \ge f(ab) = \[ \{a\}, \{b\} \],
$$
which is just the statement of the lemma. 

Now let the first two vertices on the path be $a'$ and $b'$ (in this order),
and let $A^*$ be the set of odd vertices except for $a'$, and $B^*$ be the
set of even vertices except for $b'$. Add two extra vertices, $a''$, and
$b''$ to beginning of the path. The lemma follows by induction on the 
length of the path if we show that
$$
    f(a'')+f(b'') + \[ A^*a', B^*b'\] \ge 1+ \[ A^*a'a'', B^*b'b'' \] .
$$
Now $f(a'')+f(b'') \ge f(a''b'')$, and by submodularity
$$
    \sum_{b\in B^*} f(ba'A^*) - \sum_{a\in A^*} f(a'A^*\minus a)
     \ge \sum_{b\in B^*} f(ba'a''A^*) - \sum_{a\in A^*}f(a'a''A^*\minus a) ,
$$
thus it is enough to show that
$$
   f(a''b'') + f(b'a'A^*) - f(A^*) \ge 1 + f(b'a'a''A^*) +f(b''a'a''A^*) 
          - f(a'A^*) - f(a''A^*).
$$
But this is just the sum of the following three submodular inequalities:
$$\eqalign{
   f(a''b'') -f(b'') &{}\ge 1 + f(b''a'a''A^*) - f(b''a'A^*)\cr
   f(b'')    &{}\ge f(b''a'A^*) - f(a'A^*)\cr
   f(b'a'A^*) - f(A^*) &{}\ge f(b'a'a''A^*) - f(a''A^*);
}$$
the first inequality holds as both $a''b''$ and $b''a'$ are 
edges in the graph.
\end{proof}

Now let $k$ be an even number, and let $L^d_k$ be the spanned subgraph of the
the $d$-dimensional lattice $L^d$ where only vertices with all coordinates 
between $0$ and $k$ inclusive are considered.
Thus, for example $L^d_2$ is just the $d$-dimensional cube with two
vertices along each dimension. As $L^d_k$ is a spanned subgraph of
$L^d_\ell$ whenever $k\le \ell$, the average information ratio of $L^d_k$
(not necessarily strictly) increases with $k$.
Observe also that every finite spanned subgraph of
$L^d$ is isomorphic to a spanned subgraph of $L^d_k$ for every large enough 
$k$. Thus the average information ratio of $L^d$ is the limit 
of the average information ratio of $L^d_k$ as $k$ tends to infinity. In the
sequel we estimate this latter value.

As in the proof of Theorem \ref{thm:1}, split the vertices of
$L^d_k$ into two disjoint sets $A^d_k$ and $B^d_k$ in a ``chessboard-like''
fashion so that both sets are independent, and contain just half of the
vertices: $|A^d_k| = |B^d_k| = k^d/2$.

\begin{lemma}\label{lemma:dk}
With the notation as above,
$$
   \sum_{v\in L^d_k} f(v) \ge \[ A^d_k, B^d_k \] + 
       d(k^d-k^{d-1}) - {k^d\over 2}.
$$
\end{lemma}

\begin{proof}
For $d=1$ this is the claim of lemma \ref{lemma:base}. For larger $d$
we use induction on $d$. The $(d+1)$-di\-men\-sion\-al lattice $L^{d+1}_k$
consist of just $k$ levels of $L^d_k$ with a perfect matching between the
levels. Thus we can apply lemma \ref{lemma:2} $(k-1)$ times, each
application increases the constant by the number of vertices on the new
level, i.e.~by $k^d$. Thus the constant for $(d+1)$ is $k$ times the
constant for $d$, plus $(k-1)$ times $k^d$. From here an easy calculation
finishes the proof.
\end{proof}

\begin{theorem}\label{thm:lattice}
The average information ratio of the $d$ dimensional lattice of edge length
$k$ is at least $d(1-1/k)$.
\end{theorem}

\begin{proof}
Using the notations of lemma \ref{lemma:dk}, observe that $\[ A^d_k, B^d_k
\]$ can be written as the sum of $k^d/2$ differences. Each of these
differences have value $\ge 1$ by the strong monotonicity, since the first
subset contains an edge, while the second one is independent. Thus $\[
A^d_k, B^d_k \] \ge k^d/2$. Using this, lemma \ref{lemma:dk} gives
$$
    \sum_{v\in L^d_k} f(v) \ge d(k^d-k^{d-1}).
$$
As there are $k^d$ vertices in $L^d_k$, the claim of the theorem follows.
\end{proof}

Setting $k=2$ here, we get, as a special case, that the average information
ratio of the $d$-dimensional cube is at least $d/2$. This was the hard part of
Theorem \ref{thm:1}.

Now we can finish the proof of Theorem \ref{thm:2}. We have seen that $d$ is
an upper bound for the worst case information ratio of the $d$-dimensional
lattice $L^d$. In Theorem \ref{thm:lattice} we gave the lower bound
$(d-d/k)$ for the graph $L^d_k$, which can be embedded as a spanned subgraph
into $L^d$. Thus the average information ratio of $L^d$ is larger than or
equal to the supremum of $(d-d/k)$ as $k$ runs over the even integers. 
Thus $d$ is $\le$ the average information ratio of $L^d$, which is $\le$ the
worst case information ratio of $L^d$, which is $\le d$. Thus all these values are
equal, which proves the theorem.
\end{proof}

\section{Conclusion}\label{s:conclusion}

Determining the exact amount of information a participant must remember in a
perfect secret sharing scheme is an important problem both from theoretical
and practical point of view. Access structures based on graphs pose special
challenges. They are easier to define, and have a more transparent structure
compared to general access structures. 
Research along this line poses challenges, see \cite{seven}.
Developing a new technique, we determined the {\em exact} information ratio
of the $d$-dimensional cube to be $d/2$. 
Previously this value was known to be between $d/4$ and $(d+1)/2$.

We also determined the information ratio of the (infinite) $d$-dimensional 
lattice, which turned out to be $d$.
During the proof we
estimated the information ratio of the ``finite'' lattice cube $L^d_k$ 
which has exactly $k$ vertices along each dimension. While the estimate 
was enough to get the information ratio of the infinite lattice, the 
exact (average, or worst case) information ratio for the finite graph $L^d_k$ 
remains an open problem. 

To get a better bound for the average information ratio,
consider the following secret sharing scheme.
Use the construction of Theorem \ref{thm:2}
only inside $L^d_k$, and for the missing edges on the surface use similar 
construction but with one dimension less. In this scheme
inner vertices 
will receive a total 
of $d$ bits, while vertices on the surface will receive $1/2$ bit less. Thus 
the sum the size of all shares is
$$
    d k^d - {1\over 2}\big(k^d - (k-2)^d \big)
        \approx d k^d - d k^{d-1} ,
$$
as there are $(k-2)^d$ inside vertices in $L^d_k$. Comparing this to the 
bound in Theorem \ref{thm:lattice}, the two values are approximately 
equal, but still remains a discrepancy.

Determining the worst case information ratio of $L^d_k$ seems to be a harder 
problem.
We conjecture that for $d\ge 2$, $k\ge 4$ this value equals to $d$, 
i.e.~the average information rate for the whole infinite lattice.
This conjecture was verified for $d=2$ in \cite{csirmaz:tatramountains}.

\end{document}